\documentclass[conference]{IEEEtran}
\IEEEoverridecommandlockouts
\usepackage{latexsym}
\ifCLASSINFOpdf
\else
\fi

\usepackage[cmex10]{amsmath}

\DeclareMathAlphabet{\mathpzc}{OT1}{pzc}{m}{it}

\newtheorem{theorem}{Theorem}
\begin{document}
%

\title{A New Method for Variable Elimination in Systems of Inequations}

\author{\IEEEauthorblockN{Farhad Shirani Chaharsooghi, Mohammad Javad Emadi, Mahdi Zamanighomi and Mohammad Reza Aref }
\IEEEauthorblockA{Information Systems and Security Lab (ISSL)\\
Electrical Engineering Department, Sharif University of Technology, Tehran, Iran.\\
Email: fshirani@ee.sharif.edu, emadi@ee.sharif.edu, m\_zamani@ee.sharif.edu, aref@sharif.edu}}



%



\maketitle

\begin{abstract}
 In this paper, we present a new method for variable elimination in systems of inequations which is much faster than the Fourier-Motzkin Elimination (FME) method. In our method, a linear Diophantine problem is introduced which is dual to our original problem. The new Diophantine system is then solved, and the final result is calculated by finding the dual inequations system. Our new method uses the algorithm Normaliz to find the Hilbert basis of the solution space of the given Diophantine problem. We introduce a problem in the interference channel with multiple nodes and solve it with our new method. Next, we generalize our method to all problems involving FME and in the end we compare our method with the previous method. We show that our method has many advantages in comparison to the previous method. It does not produce many of the redundant answers of the FME method. It also solves the whole problem in one step whereas the previous method uses a step by step approach in eliminating each auxiliary variable.
\end{abstract}


%
\IEEEpeerreviewmaketitle

\section{Introduction}
	The FME Method was first introduced by Fourier \cite{six} in 1827 and was rediscovered in Motzkin's thesis in 1936 \cite{seven}.   Considering the set of inequations $Ax\leq b, A\in R^{m,n}, b \in R^n$ the method provides a scheme for  determining the existence of a real or integer solution to the given problem. The method uses a step by step algorithm, in which each step
reduces the dimension of the answer space by projecting the answer space on a hyperplane. The procedure continues until the dimension is reduced to one; if the one dimensional answer is feasible, then the solution exists. A dual mode for FME was introduced in 1972 by Dantzik \cite{eight}.

	 	The FME method is widely used in solving rate region problems since there is often the need to eliminate the auxiliary variables introduced into the problem from the resulting inequations and hence reduce the dimension of the solution space by finding its projection on a given plane. The complexity of calculation in each step is of the order of $\frac{k^2}{4}$, with $k$ being the number of constraints in the system of inequalities. This complexity makes the method unfeasible for high dimensional problems. There have been many attempts to decrease the calculation time of the method. In 1996, Kebler \cite{nine} introduced a parallel computation method which was more efficient when solving the elimination problem on a computer.

In this paper, we present a new method for eliminating auxiliary variables in inequations which is much faster than the FME method. In our method, we first introduce a linear Diophantine problem which is dual to the original problem. We then proceed by solving this new Diophantine system, next we calculate the final result by finding the dual inequations system. Our new method uses the algorithm Normaliz to find the Hilbert basis of the solution space of the given Diophantine problem.

 Normaliz was first introduced as a program in 1997 by R. Koch. In 2003, the algorithm was presented in Koch's thesis \cite{ten}. Normaliz was originally introduced to compute the lattice points in a lattice polytope and the integral closure of a monomial ideal. Koch mentions in his thesis that the algorithm may be extended to find the Hilbert basis of the solution space of a system of linear Diophantine equations. However, applying the algorithm introduced in \cite{ten} to information theoritic problems is still time-consuming. In 2010, Bruns and Ichim \cite{eleven} introduced a dual algorithm which is much faster when applied to problems with a greater number of constraints than auxiliary variables. Since this is the case in almost all information theory problems, we use this method to find the desired Hilbert Basis solution.

One of the most important fundamental channels in communication systems is the Interference Channel (IC). The IC was initiated by Shannon \cite{one}. Simple outer bounds for the IC were established by Ahlswede \cite{two}. Later, Carlieal \cite{three}, using superposition encoding and sequential decoding, characterized an achievable rate region for general discrete memoryless IC. In the IC each receiver is only interested in the message from its corresponding sender. However, it might decode messages from other senders to reduce the interference effect. Han and Kobayashi (HK) utilized rate splitting techniques in \cite{four}, and applied simultaneous superposition coding techniques to establish the largest inner bound introduced for the general IC to date. Computing the capacity region of general IC is still an open problem, but in special cases such as strong interference regime the capacity region has been established, and the region was the same as the HK rate region \cite{five}. In the method used to achieve the HK rate region, the messages are split into two parts, the private and common parts; private messages are decoded by the corresponding receiver, while the common messages are decoded by all receivers.

 Rate splitting has become a prominent tool for solving rate region problems. However, there is a drawback to the method in that solving the channel with the new split messages yields a system of inequations based on the split-rates. Since we are interested in each sender's total rates, the auxiliary variables must be eliminated. The FME method is needed to eliminate these auxiliary rates.

The rest of the paper is organized as follows: In Section II, we first introduce our model of an $l$ transmitter-receiver interference channel. In Section III, we introduce a problem in the IC. In this section, we use some simplifying assumptions, since the goal of the section is to explain the new elimination method. In Section IV, we solve the problem with our new method. In Section V, we present a generalized algorithm for our new method which could be applied to any problem involving FME. In Section VI, we present results of applying our new method to several problems which were solved previously using the FME method. Section VII concludes the paper.




\section{Channel Moddel}
In this section, we introduce our model of IC. We denote random variables by upper case letters e.g. X, their realizations are presented by lower case letters e.g. x in finite sets denoted by $\mathpzc{X}$. We use $x^k_i$ to denote the vector $x_i,x_{i+1}, ... , x_k$. We denote $x^k$ if i=1, while for i=1 and k=$l$, we show the vector by $\underline{x}$. Also, if i=1 and k=$2l$, we show the vector by \underline{\underline{x}}. Furthermore $\underline{x}_i^n$ is the compact form of $(\underline{x}_i,\underline{x}_{i+1},...,\underline{x}_n)$ where $\underline{x}_i=(x_{i1},...,x_{il})$.

A discrete Interference Channel, with $\mathit{l}$ senders and $\mathit{l}$ receivers is the 2$\mathit{l}$+1 tuple
$(\underline{\mathpzc{X}},p(\underline{y}|\underline{x}),\underline{\mathpzc{Y}})$
where $\mathpzc{X}_i$ are the input alphabets, $\mathpzc{Y}_i$ are the output alphabets for i=1,2,...,$\mathit{l}$ and $p(\underline{y}|\underline{x})$ is the conditional probability of the channel. Each decoder is only interested in decoding its corresponding sender's message; however, in order to reduce the interference effect, receivers might decode the message from other senders as well. The channel is memoryless and time invariant in the sense that
\begin{equation*}
p(\underline{y}_n|\underline{x}^n,\underline{y}^{n-1})=p_{\underline{Y}|\underline{X}}(\underline{y}_n|\underline{x}_n).
\end{equation*}
A $((2^{nR_1},2^{nR_2},...,2^{nR_\mathit{l}}),n)$ code consists of (i) $\mathit{l}$ message sets $\mathpzc{W}_i={1,2,...,2^{nR_i}},i\in[1:\mathit{l}]$, (ii) $\mathit{l}$ encoding functions $f_u=\mathpzc{W}_u\rightarrow \mathpzc{X}^n_u, u\in[1:\mathit{l}]$, (iii)$\mathit{l}$ decoding functions $g_u(.)$ where $\hat{W}_u=g_u(Y_u^n),u\in[1:\mathit{l}]$. We define the average probability of error for this code as $P_e^{(n)}=Pr\{\bigcup_{i=1}^\mathit{l}(\hat{W_i}\neq W_i)\}$. A nonnegative rate $\mathit{l}$-tuple $(R_1,R_2,...,R_\mathit{l})$ is said to be achievable if there exists a code $((2^{nR_1},2^{nR_2},...,2^{nR_\mathit{l}}),n)$ for which $P_e^{(n)}\rightarrow 0$ as $n$ tends to infinity.
\section{Problem Statement}
In this section, we present our coding strategy which is a generalization of the method used in \cite{four}. We describe the problem in solving the resulting inequations with the conventional FME method. After stating the problem, we introduce our method and compare the results with the former procedure.

Since our goal in this paper is to introduce an alternative elimination method for the FME method, we use some simplifying assumptions in the channel. We consider the channel to be symmetric in the sense that the $i$th sender sends the same amount of information to all receivers other than its own corresponding receiver. This allows each sender to send the same common message to all receivers. We split each message $W_i$ into common and private parts $(W_{ic},W_{ip})$. The common part of the message is decoded by all receivers while the private part is decoded only by the $i$th receiver. Another suitable assumption would have been to assume a sense of degradedness between the receivers in the following way: assume that receiver $i_1$ receives all the common massages sent by sender $i$; receiver $i_2$ receives all the common messages received by the receivers other than $i_1$, and so on. In this case we had to split the message in each sender to $\mathit{l}$ common parts and one private part in the following way: split the message $W_i$ into $(W_{i(1c)},W_{i(2c)},...,W_{i(\mathit{l}c)},W_{ip}$). $W_{i(\mathit{l}c)}$ is decoded by the receiver with the least information about the massage from sender $i$; the reciever with the second least information about this message decodes $W_{i(\mathit{l}c)},W_{i(\mathit{(l-1)}c)}$, and so on. The rest of the solution would stay the same. Clearly, the symmetric case mentioned above is a special case of this latter case. Since the two cases have the same solution, and for the sake of simplicity, we solve the channel for the symmetric case in this paper.
\\\textbf{Codebook Generation:}

 Let $Z_1=(U_1,U_2,...,U_\mathit{l},X_1,X_2,...,X_\mathit{l},Y_1,Y_2,...,Y_\mathit{l},Q)$, where Q is the time-sharing variable. Let $\mathpzc{P}_1$ denote the set of all jont p.m.fs $p_1$ on $Z_1$, which can be written in the form
\begin{equation*}
p(z_1)=p(q)\prod_{i=1}^l (p(u_i|q)p(x_i|u_i,q))p(y_1,y_2,...,y_l|x_1,x_2,...,x_\mathit{l})
\end{equation*}
For any $p_1$ defined above, generate $2^{nR_{ic}},i\in[1:l]$, i.i.d $u^n_{ic}(w_{ic})$ where $w_{ic}\in[1:2^{nR_{ic}}]$. For each $u^n_{ic}(w_{ic})$, generate $2^{nR_{ip}}$, i.i.d $x^n_{i}(w_{ip},w_{ic})$. The time-sharing variable has no impact on the rest of the proof and for the sake of brevity, we ignore it until the end of the proof where we bound its cardinality.
\\\textbf{Encoding:}

At the beginning of each block, the $i$th sender picks related $(u^n_{ic},x^n_{i})$ according to $\prod^n_{j=1} p(u_{ic,j})p(x_{i,j}|u_{ic,j})$ and sends it.
\\\textbf{Decoding:}

At the end of each block, the $i$th receiver finds $(x^n_i,u^n_{1c},u^n_{2c},...,u^n_{lc})$ such that:
\begin{equation*}
(x^n_i(w_{ip},w_{ic}),u^n_{1c}(w_{1c}),u^n_{2c}(w_{2c}),...,u^n_{lc}(w_{lc}))\in A_\epsilon
\end{equation*}
Without loss of generality we consider the case where all senders have sent index 1 as their messages. We investigate all possible scenarios in the first decoder. The decoder must decode messages $(w_{1p},w_{1c},w_{2c},...,w_{lc})$, the decoding is errorless if $\hat{w}=(\hat{w}_{1p},\hat{w}_{1c})=(1,1)$. Using the same method as in \cite{twelve} and by the AEP rule, it is clear that the probability of error approaches zero as $n\rightarrow\infty $ (e.g. rate $l$-tuple is achievable) iff the rates satisfy the following set of equations:
\begin{equation}
R_{1p}+\alpha_1R_{1c}+\alpha_2R_{2c}+...+\alpha_lR_{lc}\leq I_{1,\alpha_1,\alpha_2,...,\alpha_l}
\label{eq}
\end{equation}
where
\begin{equation*}
\alpha_i\in\{0,1\}
\end{equation*}

\begin{equation*}
I_{1,\alpha_1,\alpha_2,...,\alpha_l}=I(X_1,U_S;Y_1|U_{S^c})
\end{equation*}

and $U_i\in U_S$ if $\alpha_i=1$; also, for $\alpha_i=0$, we have $U_i\in U_{S^c}$.
We note that $R_1-R_{1c}=R_{1p}$, substituting into \eqref{eq} we have
\begin{equation*}
R_1-\alpha_1R_{1c}+\alpha_2R_{2c}+...+\alpha_lR_{lc}\leq I_{1,\alpha_1,\alpha_2,...,\alpha_l},\alpha_i\in\{0,1\}
\end{equation*}
These are the inequations for decoder 1. We rewrite the inequations for all decoders in matrix form:
\begin{equation}
\left (
\begin{matrix}
1 & 0 & \ldots & 0 & 0\\
1 & 0 & \ldots & 0 & 0\\
\vdots & \vdots & \ddots & 0 & 0\\
1 & 0 & \ldots & 0 & 0\\
0 & 1 & \ldots & 0 & 0\\
0 & 1 & \ldots & 0 & 0\\
\vdots & \vdots & \ddots & 0 & 0\\
0 & 1 & \ldots & 0 & 0\\
\vdots & \vdots & \ddots\\
\vdots & \vdots & \ddots\\
0 & 0 &\ldots & 0 & 1 \\
0 & 0 &\ldots & 0 & 1 \\
\end{matrix} \right |
\left .
\begin{matrix}
\\
& & A_1  \\
 \\
\\
\_&\_&\_&\_\\
\\
  & &   A_2  \\
\\
\_&\_&\_&\_\\
\vdots & \vdots& \vdots & \ddots\\
\_&\_&\_&\_\\
\\
  & &   A_l \\

\\
\end{matrix} \right)
\left(
\begin{matrix}
R_{1}\\
R_{2}\\
\\
\vdots\\
\\
\vdots\\
R_{l}\\
R_{1c}\\
R_{2c}\\
\vdots\\
\\
\vdots\\
\\
R_{lc}\\
\end{matrix} \right)
\leq
\left(
\begin{matrix}
I_{1,00...0}\\
I_{1,00...1}\\
\vdots\\
I_{1,11...1}\\
I_{2,00...0}\\
I_{2,00...1}\\
\vdots\\
I_{2,11...1}\\
\vdots\\
\vdots\\
I_{l,00...0}\\
\\
\\
I_{l,11...1}\\
\label{eqq}
\end{matrix} \right)
\end{equation}
If we denote A as the matrix whose rows are all possible binary $l$-tuples starting from (0, 0, ..., 0) in an increasing order, then $A_i$ is the matrix formed by multiplying the $i$th column in A by -1. We show the rate coefficients matrix in \eqref{eqq} with $C_{l}$. We show the rate vector with $\underline{\underline{R}}^{t}$, and the mutual information vector with $I_l$, so we may write \eqref{eqq} in compact form as $C_{l}\underline{\underline{R}}^{t}\leq I_l$.
To clarify the method of forming the above mentioned inequations, we give the following example, which illustrates the inequations for a two sender-receiver interference channel:
\begin{equation*}
C_{2}\underline{\underline{R}}^{t}\leq I_2 \Rightarrow
\left (
\begin{matrix}
1 & 0 \\
1 & 0 \\
1 & 0 \\
1 & 0 \\
0 & 1 \\
0 & 1 \\
0 & 1 \\
0 & 1 \\
\end{matrix} \right |
\left .
\begin{matrix}
0 & 0 \\
0 & 1 \\
-1 & 0 \\
-1 & 1 \\
0 & 0 \\
0 & -1 \\
1 & 0 \\
1 & -1 \\
\end{matrix} \right)
\left(
\begin{matrix}
R_{1}\\
R_{2}\\
R_{1c}\\
R_{2c}\\
\end{matrix} \right)
\leq
\left(
\begin{matrix}
I_{1,00}\\
I_{1,01}\\
I_{1,10}\\
I_{1,11}\\
I_{2,00}\\
I_{2,01}\\
I_{2,10}\\
I_{2,11}\\
\end{matrix} \right)
\end{equation*}
Note that these are the same inequations given in \cite{twelve}. To solve this system with the conventional FME method, one must eliminate all of the variables $R_{ic}, i\in[1:l]$; this process takes $l$ rounds of applying FME to the matrix. According to \cite{twelve}, the size of the matrix grows at a worst case rate of $\frac{k^2}{4}$ each round. This complexity renders the aforementioned method impractical for $l$ larger than 2, since, after applying FME, one must eliminate the redundant inequations in the resulting matrix using an exhaustive search. Finding the redundant inequalities often needs more time than eliminating the auxiliary variables, which itself is a time-consuming process.
\section{Our New Method}

 Here, we present a new method for solving \eqref{eqq} which not only requires much less calculation in comparison with the FME but also produces no redundant inequations assumning the mutual-informations on the right side of the inequations are independant. Before proceeding with the solution, we present some definitions. We show the matrixes resulting from $l$ rounds of FME by $C^*_l$ and $I^*_l$, so the resulting inequations can be written in the form $C^*_l\underline{\underline{R}}^t\leq I^*_l$. Note that the second $l$ columns of $C^*_l$ are zero. It is also clear that if we show the rows of $C_l$ by $C_{l(l,\alpha_1,\alpha_2,...,\alpha_l)}$, then all rows of $C^*_l$ and $I^*_l$ can be shown by
\begin{equation*}
C^*_{l(k)}=\sum_{i,\alpha_1,\alpha_2,...,\alpha_l}a_{k(i,\alpha_1,\alpha_2,...,\alpha_l)} C_{l(i,\alpha_1,\alpha_2,...,\alpha_l)},
\end{equation*}
\begin{equation*}
I^*_{l(k)}=\sum_{i,\alpha_1,\alpha_2,...,\alpha_l}a_{k(i,\alpha_1,\alpha_2,...,\alpha_l)} I_{l(i,\alpha_1,\alpha_2,...,\alpha_l)},
\end{equation*}
where
\begin{equation*}
\alpha_i\in\{0,1\},i\in[1:l],a_{k(i,\alpha_1,\alpha_2,...,\alpha_l)} \geq 0
\end{equation*}
Let $G_{C_l}$ be the solution space of the inequations in \eqref{eqq} for rates $R_1,R_2,...,R_l$, and let $G_{C^*_l}$ be the solution space of $C^*_l\underline{\underline{R}}^t\leq I^*_l$. Furthermore, we define the matrix D as the matrix whose rows are formed by all of the linear combinations of $C_{l(i,\alpha_1,\alpha_2,...,\alpha_l)}$ with positive coefficients and with the property that the second $l$ columns of D are zero. In other words, $\underline{\underline{D}}$ is a row of D iff
\begin{equation}
\underline{\underline{D}}=\sum_{i,\alpha_1,\alpha_2,...,\alpha_l}a_{i,\alpha_1,\alpha_2,...,\alpha_l} C_{l(i,\alpha_1,\alpha_2,...,\alpha_l)}
\label{D}
\end{equation}
with the following conditions:
\begin{equation*}
a_{k(i,\alpha_1,\alpha_2,...,\alpha_l)} \geq 0
 \end{equation*}
\begin{equation}
D_{l+1}^{2l}=0;
\label{column}
 \end{equation}
We define $G_{D_l}$ as the solution space of the following inequations:
\begin{equation}
D_l\underline{\underline{R}}^t\leq I_{D_l}
\label{eqqq}
\end{equation}
where
\begin{equation*}
I_{D_l}=\sum_{i,\alpha_1,\alpha_2,...,\alpha_l}a_{i,\alpha_1,\alpha_2,...,\alpha_l} I_{l(i,\alpha_1,\alpha_2,...,\alpha_l)}\end{equation*}
\begin{theorem}
$G_{C^*_l}$ is equal to $G_{D_l}$.
\end{theorem}
\begin{proof}The proof is straightforward. Since all rows of $D_l$ are linear combinations of rows of $C_l$ with positive coefficients, all inequations in \eqref{eqqq} result from inequations in \eqref{eqq}. Hence, we have $G_{C_l}\subseteq G_{D_l}$, and by the FME, we know that $G_{C_l}$ is equal to $G_{C^*_l}$ for $(R_1,R_2,...,R_l)$, we have $G_{C^*_l}\subseteq G_{D_l}$. On the other hand, since all the rows of $C^*_l$ are linear combinations of rows of $C_l$ with the property that for each row $\underline{\underline{C}}^*$, ${C^*}^{2l}_{l+1}=0$, any row $\underline{\underline{C}}^*$ of $C^*_l$ is a row of $D_l$, and so $G_{D_l}\subseteq G_{C^*_l}$ and the proof is complete.
\end{proof}

Now, we propose that instead of finding $G_{C^*_l}$, we try to find $G_{D_l}$. From \eqref{column}, we know that for every \underline{\underline{D}} in \eqref{D} we have
\begin{align}
&\sum_{\substack{i,\alpha_1\alpha_2,...,\alpha_{j-1},\alpha_{j+1},...,\alpha_l \\ i \neq j}}a_{i,\alpha_1,\alpha_2,...,\alpha_{j-1},1,\alpha_{j+1},...,\alpha_{l}}\nonumber \\
&= \sum_{\alpha_1,\alpha_2,...,\alpha_{j-1},\alpha_{j+1},...,\alpha_l }a_{j,\alpha_1,\alpha_2,...,\alpha_{j-1},1,\alpha_{j+1},...,\alpha_{l}}
\label{So}
\end{align}

Now we need to choose inequations in \eqref{eqqq} whose linear combinations with positive coefficients produce all inequations in \eqref{eqqq}. Without loss of generality, we consider $a_{j,\alpha_1,\alpha_2,...,\alpha_{l}}\in Z^+$. We give the following definition as in \cite{fifteen};
\\\textbf{Definition:} We call a set of vectors H, generator vectors of S, if every vector in S can be expressed as a linear combination of vectors in H with positive coefficients and no vector in H may be expressed in such a way by other vectors in H. Such a set of vectors is called the Hilbert Basis of the vector space S.

The problem of finding a Hilbert Basis for the solutions of a Diophantine system of equations such as \eqref{So} has many solutions. One solution is to use an algorithm called Normaliz introduced by Koch in \cite{ten}. We use this method to solve \eqref{So}. If the set of $a_{j,\alpha_1,\alpha_2,...,\alpha_{l}}$, $j=1,...,l$,$\alpha_i\in\{0,1\}$ is an answer for \eqref{So}, then we have the following inequation for $R_1,...,R_l$:
\begin{equation*}
\sum_{j=1}^l ((\sum_{\alpha_i} a_{j,\alpha_1,\alpha_2,...,\alpha_{l}})R_j)\leq \sum_{\alpha_i,j} a_{j,\alpha_1,\alpha_2,...,\alpha_{l}}I_{j,\alpha_1,\alpha_2,...,\alpha_l}
\end{equation*}
Since Normaliz provides a Hilbert Basis solution for the equations in \eqref{So}, and since each solution for \eqref{So} represents a unique equation in \eqref{eqqq} and adding two answers from \eqref{So} with positive coefficients is equivalent to adding two inequations in \eqref{eqqq} with positive coefficients, the resulting inequations contain all the non-redundant inequations in \eqref{eqqq}. Hence, by applying Normaliz, we can calculate the solution to our original problem. Since the dimension of the Diophantine space in \eqref{So} is much greater than the number of constraints we have, it is better to use Normaliz Dual as suggested in \cite{eleven}. Our results for the two sender-receiver IC are the same as HK rate region in \cite{four}.

 To this point, we have found the rate region by eliminating the auxiliary variables. Since, in the case of the above channel, all the mutual-informations $I_{l(i,\alpha_1,\alpha_2,...,\alpha_l)}$ are independent of each other  (i.e. they don't have a clear relation with each other), there are no redundant inequations in our final system of inequations and so by the Caratheodory theorem \cite{sixteen} Q is bounded by the number of resulting inequations after applying Normaliz and finding the dual system.
\section{Generalization to Other Problems}
The method proposed above may be used in any other problem which could be solved using the FME. The following algorithm is used:
\\1- Define the encoding and decoding schemes.
\\2- Find the inequations matrix corresponding to the defined schemes.
\\3- Write the Diophantine equations of the matrix.
\\4- Use Normaliz to solve the Diophantine problem.
\\5- For each answer vector for the Diophantine problems, write the corresponding inequation.

\section{Simulation and Results}
Here, we use our method to solve some problems which have been solved using the FME method previously. The simulation results are illustrated in \eqref{tab:SystemParameters}:
\begin{table}[ht]
	\centering
	\caption{Results of Applying Our New Method on Some Examples}
		\begin{tabular}{| c | c | c | c | c | c |}\hline
Name& Number of&Number of&Basis&Non-Redundant&t\\&constraints&  aux. variables &Elements&Elements&(sec)\\\hline
		HK &8&2&7&7&$\leq 1s$\\l=2&&&&&\\\hline
		HK &24&3&153&153&$\leq 1s$\\l=3&&&&&\\\hline
		HK &64&4&56384&56384&203s\\l=4&&&&&\\\hline
		Ref. 15&29&6&286&60& $\leq$1s\\&&&&&\\\hline
		Ref. 16&12&5&6&3&$\leq$1s\\&&&&&\\\hline			
\end{tabular}		
	\label{tab:SystemParameters}
\end{table}

 The results are obtained from running Normaliz 2.2 on a 2.4 i5 CPU. The 4th column indicates the number of inequations resulting from applying our method while the 5th column indicates the number of non-redundant constraints. Note that in cases such as \cite{thirteen} and \cite{fourteen} where there are certain relations between mutual-informations on the right side of the inequations, there may be redundant answers in our method. However all these redundant answers will be present in the FME case as well. Table \eqref{tab:SystemParameters} shows the efficiency of our method in comparison to the FME method. There are several features of our method that make it more efficient in comparison with the former method. One feature is that the FME method eliminates one variable at a time, while our method eliminates all variables in one step. Another feature is that our method is much faster compared to the previous method. Another important advantage of the new method is that it doesn't introduce many of the redundant inequation present in the FME solution. Note that one of the most time-consuming parts of the FME method is the removal of inequations which are linear combinations of previous inequations with positive coefficients. This process uses an exhaustive search and is very inefficient. Our method doesn't introduce any such redundant answers.

\section{Conclusion}
 In this paper, we introduced a new method for performing variable elimination in systems of inequations. We illustrated the implementation of our new method on a special case of IC. A generalization of the method was presented and in the end some simulation results were compared with the actual results after applying FME and removing redundant answers. It was shown that our new method is much faster and more efficient than the FME method. Our method eliminates all variables in one step while the FME method uses a step by step approach. Also, our method produces fewer redundant inequations in comparison to the FME method.


\section*{Acknowledgments}

We would like to express our appreciation for the ISSL Group of Sharif
University of Technology and specially B. Akhbari and M. Mirmohseni for their
valuable comments.




\begin{thebibliography}{1}

\bibitem{six}
J-B.J. Fourier, reported in: Analyse des travaux de l`Acadamie Royale
des Sciences, pendant l`annee 1824, Partie mathematique, \emph{Histoire de
l`Academie Royale des Sciences de l`Institut de France} 7 (1827) xlvii–
lv. (Partial English translation in: D.A. Kohler, Translation of a Report by
Fourier on his work on Linear Inequalities, Opsearch 10 (1973) 38–42.)
\bibitem{seven}
T.S. Motzkin. Beitr\"{a}ge zur Theorie der linearen Ungleichungen. (Inangural Dissertation Basel). Azriel, Jerusalem, 1936. English translation: Contributions to the theory of linear inequalities, RAND Corporation Translation 22, Santa Monica, CA, 1952. Reprinted in \emph{Theodore S. Motzkin: Selected Papers} (D. Cantor, B. Gordon, B. Rothschild. eds.), Birkh\"{a}user, Boston, 1983, pp. 1-80.
\bibitem{eight}
G.B. Dantzig, and B.C. Eaves, \emph{Fourier-Motzkin Elimination and its Dual},.\hskip 1em plus
  0.5em minus 0.4em\relax
Journal of Combinatorial Theory Ser. A, 14:288–297, 1973.
\bibitem{nine}
C. W. Kebler, \emph{Parallel fourier-motzkin elimination}, In Euro-Pat, Vol. II,
pp. 66–71, 1996.
\bibitem{ten}
R. Koch, Affine Monoids, Hilbert Bases and Hilbert Functions. PhD thesis,  Osnabrück university, Germany, 2003.
\bibitem{eleven}
 W. Bruns and B. Ichim. \emph{Normaliz: algorithms for affine monoids and rational cones}, J. Algebra 324 (2010), no. 5, 1098 – 1113.


\bibitem{one}
C. E. Shannon, \emph{Two-way communication channels}, in Proc. 4th Berkeley Symp. on
Mathematical Statistics and Probability, vol. 1, Berkeley, CA: Univ. California
Press, 1961.
\bibitem{two}
R. Ahlswede, \emph{The capacity region of a channel with two senders and two receivers}, Annals Probabil., vol. 2, no. 5, pp. 805-814, 1974.
\bibitem{three}
A. B. Carleial, \emph{Interference channels}, IEEE Trans. Inform. Theory, vol. IT-24, no. 1, pp. 60–70, Jan. 1978.
\bibitem{four}
T. S. Han and K. Kobayashi, \emph{A new achievable rate region for the interference channel},.\hskip 1em plus
  0.5em minus 0.4em\relax
IEEE Trans. Inf. Theory, vol. 27, no. 1, pp. 49-60, 1981.
\bibitem{five}
R. Etkin, D. Tse, and H. Wang, \emph{Gaussian interference channel capacity to within one bit}, IEEE Trans. Inf. Theory, vol. 54, no. 12, pp. 5534–5562, 2008.

\bibitem{twelve}
A. El Gamal and Y.-H. Kim,\emph{Lecture notes on network information theory}
.\hskip 1em plus
  0.5em minus 0.4em\relax
http://arxiv.org/abs/1001.3404, Jan. 2010.

\bibitem{fifteen}
W. Bruns and J. Gubeladze. Polytopes, rings, and K-theory. Springer Monographs in Mathematics, 2009.
\bibitem{sixteen}
C. Carathéodory, Über den Variabilitätsbereich der Fourierschen Konstanten von positiven harmonischen Funktionen, Rend. Circ. Mat. Palermo, Vol. 32, pp. 193-217, 1911.

\bibitem{thirteen}
M. Mirmohseni, B. Akhbari, and M. R. Aref, \emph{On the capacity of causal
cognitive interference channel with delay}, submitted to IEEE Trans. Inform. Theory, April 2010.
\bibitem{fourteen}
Y. K. Chia, and A. El Gamal, \emph{3-Receiver broadcast channels with common and confidential messages},
IEEE Int. Symp. Information Theory (ISIT 2010), Seoul, Korea, pp. 1849 - 1853, June 28-July 3, 2009.

\end{thebibliography}
%

\end{document}